\title{A tropical geometry approach to BIBO stability } 
\author{B. Bossoto \and M. Mboup \and A. Yger}

\documentclass[12pt]{article}
\usepackage{graphicx}
\usepackage{amssymb}
\usepackage{amsfonts}
\usepackage{graphicx}
\usepackage{amsmath}
\usepackage[english]{babel}
\usepackage{mathrsfs}
\newtheorem{theorem}{Theorem}

\newtheorem{definition}{Definition}

\newtheorem{proposition}{Proposition}
\newtheorem{remark}{Remark}

\newenvironment{proof}[1][Proof]{\textbf{#1.} }{\ \rule{0.5em}{0.5em}}

\usepackage[left=3cm,
right=3cm,top=1.5cm,bottom=2cm]{geometry}

\def\C{\mathbb C}

\def\R{\mathbb R}
\def\F{\mathbb F}

\def\N{\mathbb N}
\def\Q{\mathbb Q}
\def\P{\mathbb P}
\def\Z{\mathbb Z}
\def\T{\mathbb T}

\def\D{\mathbb D}

\def\bfA{{\boldsymbol A}}

\def\({{\rm (}}
\def\){{\rm )}}

\def\0{\underline{0}}

\begin{document}
\newpage
\maketitle
\begin{abstract}
Given a Laurent polynomial $F\in \C[X_1^{\pm 1},...,X_n^{\pm 1}]$, its {\it am\oe ba} $\mathcal A_F$ is the image by $z=(z_1,...,z_n)\in (\C^*)^n 
\longmapsto (\log |z_1|,...,\log |z_n|)\in \R^n$ of the algebraic zero set 
$V(F)=\{z\in (\C^*)^n\,;\, F(z)=0\}$ of the complex torus $\T^n :=(\C^*)^n$. We relate here the question of the BIBO {\it stability of a multilinear time invariant system} with transfer function $A(X^{\pm 1})/B(X^{\pm 1}) = G(X_1,...,X_n)/F(X_1,...,X_n)$, where $F, G\in \C[X_1,...,X_n]$ are coprime in $\C[X_1,...,Xn]$, with the geometrical study of the am\oe ba $\mathcal A_F$. We formulate very simple criteria for BIBO {\it strong} or {\it weak stability} in terms of the position of $\underline 0=(0,...,0)\in \R^n$ with respect to the am\oe ba $\mathcal A_F$ and suggest an algorithmic procedure in order to test such property when $F\in \Z[X_1,...,X_n]$. Such procedure relies on the concept of {\it lopsidedness approximation of $\mathcal A_F$}, as introduced by K. Purbhoo \cite{Purb08} and completed from the algorithmic point of view in \cite{FMMW19}. 
\end{abstract}

\section{Introduction} 

Let $n\in \N$ and 
$$
\C^{[\Z^n]}:= 
\{(u_k)_{k\in \Z^n} \in \C^{\Z^n}\,;\, u_k =0\ 
{\rm except\ for\ a\ finite\ number\ of}\ k\in \Z^n\}. 
$$
A discrete linear time-invariant system 
$$
(u_k)_{k\in \Z^n} \in \C^{[\Z^n]} \stackrel{S}{\longmapsto} 
\Big( \sum\limits_{\kappa \in \Z^n} h_\kappa u_{k-\kappa}\Big)_{k\in \Z^n} \in \C^{\Z^n} 
$$
is said to be {\it Bounded Input- Bounded Output} (BIBO) {\it stable} if and only if its impulse response $(h_k)_{k\in \Z^n}$ belongs to $\ell^1_\C(\Z^n)$, that is 
\begin{equation}\label{sect1-eq1} 
\sum\limits_{k\in \Z^n} |h_k| < +\infty. 
\end{equation} 
Since $\ell^1_\C(\Z^n) \hookrightarrow \ell^2_\C(\Z^n)$, \eqref{sect1-eq1} implies that $\sum_{k\in \Z^n} |h_k|^2< +\infty$, which corresponds to the fact that the system $S$ is {\it asymptotically stable} (or {\it stationary}). BIBO stability thus constitutes a stronger requirement than just {\it asymptotic stability}.  
\vskip 2mm
\noindent
There exists in the classical literature several criteria for BIBO stability for discrete linear time-invariant systems (see for example  \cite{Huang72, JSh73, AndJ74, Sch76, Bose77, Good77, Bist99, BeDir99}).
\vskip 2mm
\noindent 
In this paper, we are concerned with  discrete multilinear time invariant 
systems which admit as transfer function the rational function 
$$
\frac{B(X_1^{-1},...,X_n^{-1})}{A(X_1^{-1},...,X_n^{-1}} = X^\gamma \frac{G(X_1,...,X_n)}{F(X_1,...,X_n)}\in \C[X_1,...,X_n], 
$$ 
where 
$\gamma \in \Z^n$ and $G,F\in \C[X_1,...,X_n]$ are coprime in $\C[X_1,...,X_n]$, both $F$ and $G$ being also coprime with $X_1\cdots X_n$. They are called {\it discrete $n$-rational filters}. In case the rational function 
$$
z\in \T^n=(\C^*)^n 
\longmapsto G(z)/F(z),
$$
where $\C^* := \C\setminus \{0\}$, is regular about the $n$-dimensional torus 
$$
\mathbb S_1^n = \{z=(e^{i\theta_1},...,e^{i\theta_n})\,;\, \theta \in (\R/(2\pi\Z))^n\},
$$ 
then the two following assertions are equivalent for such a discrete $n$-rational filter $S$~: 
\begin{equation*} 
\begin{split} 
&(i)\ S \ {\rm is\ BIBO\ stable} \\
&(ii)\  
F^{-1} (\{0\}) \cap \{z\in \C^n\,;\, |z_j|\leq 1\ {\rm for}\ j=1,...,n\} = \emptyset.  
\end{split} 
\end{equation*} 
Such an equivalence is known in the case where $n=1$ or $n=2$ as {\it Shanks criterion} \cite{JSh73}. 
In the case where $n=2$, one may formulate various criteria equivalent to this one. Here are two examples. 

\begin{theorem}[\cite{Huang72}]\label{sect1-thm1}   
Under the condition that $G/F$ is regular 
about $\mathbb S_1^n$, it is equivalent to say that the discrete  
$2$-rational filter $S$ with transfer function 
$$
\frac{B(X_1^{-1},X_2^{-1})}{A(X_1^{-1},X_2^{-1})} = X^\gamma 
\frac{G(X_1,X_2)}{F(X_1,X_2)}
$$
is 
BIBO stable and that 
\begin{equation*} 
F(z_1,0) \not=0 \quad {\rm for}\ |z_1|\leq 1,\quad 
F(z_1,z_2) \not=0\quad {\rm for}\, |z_1|=1, |z_2|\leq 1. 
\end{equation*}    
\end{theorem}    

\begin{theorem}[\cite{BeDir99}] Suppose that $\deg_{X_1} F=d_1$ and $\deg_{X_2} F=d_2$. Let $F^*$ be the conjugate polynomial such that 
$$
F^*(z_1,z_2) := z_1^{d_1} z_2^{d_2} \, \overline{F\Big( 
\frac{1}{\bar z_1},\frac{1}{\bar z_2}\Big)}\quad ((z_1,z_2)\in \T^2)  
$$
and $R_{X_2}(X_1)$ be the resultant of $F$ and $F^*$ considered as elements of $\C[X_1] [X_2]$. Under the condition 
that $G/F$ is regular in $(\mathbb S^1)^n$, 
it is equivalent to say that the discrete $2$-rational 
filter $S$ with transfer function 
$$
\frac{B(X_1^{-1},X_2^{-1})}{A(X_1^{-1},X_2^{-1})} = X^\gamma 
\frac{G(X_1,X_2)}{F(X_1,X_2)}
$$
is 
BIBO stable and that 
\begin{eqnarray*} 
F(z_1,0) &\not=0& \quad    {\rm for}\ |z_1|\leq 1 \\ 
F(1,z_2) &\not=0& \quad {\rm for}\, |z_2|\leq 1\\
R_{X_2}(z_1) &\not =0 &\quad {\rm for}\, |z_1|=1.   
\end{eqnarray*}      
\end{theorem} 
\vskip 2mm
\noindent 
Such criteria consist in formulating the BIBO stability condition 
in the two dimensional setting in such a way one it can be tested 
thanks to one dimensional tests of the Schur-Cohn type which are well known. Tests and algorithms issued from them need a lot of computations.
As an example, consider the Bose test \cite{Bose77}. It consists in reducing the test proposed in Theorem \ref{sect1-thm1}  
to four one-dimensional tests, thus inducing a heavy computational machinery.
\vskip 2mm
\noindent 
In the higher multidimensional case, alternative methods have been introduced by M. Najim, I. Serban and F. Turcu. Such methods are based on the introduction of so-called {\it Schur coefficients families} in several variables, the goal being to obtain a multidimensional Schur-Cohn criterion (see \cite{ANRT03, NST06, NS07A, NS07B, NS07C}).
Let 
$$
F(X_1,...,X_n) = \sum\limits_{\alpha \in {\rm Supp}\, F\subset \Z^n} 
c_\alpha \, X^\alpha\quad (c_\alpha \in \C^*)
$$
with total degree $\delta_F$ in the $n$ variables $X_1,...,X_n$. 
For any $w=(e^{i\omega_1},...,e^{i\omega_{n-1}})$ with 
$\omega = (\omega_1,...,\omega_{n-1}) \in (\R/(2\pi\Z))^{n-1}$, 
define $F_w\in \C[Y]$ by
$$
F_w(Y) = F(w_1 X_1,...,w_{n-1} X_{n-1}, Y) = F(e^{i\omega_1} X_1,...,e^{i\omega_{n-1}} X_{n-1},Y) \in \C[Y].  
$$
Let also $F_w^*(Y)$ be the polynomial defined by 
$$
F_w^*(u) = u^{\delta_F} \overline{F_w (1/\bar u)}. 
$$
In the sequel, the open (respectively closed) unit disc of the complex plane is de\-noted $\D = \{z
\in \C\ ; |z| < 1\}$ (resp. $\overline{\D} = \{z
\in \C\ ; |z| \leq 1\}$).

The main result they obtain is that the three following statements are equivalent:   
\begin{itemize}
\item $F(z)\not=0$ for any $z\in \overline \D^n$ ;
\item for any $w=(e^{i\omega_1},...,e^{i\omega_{n-1}})$, $F_w(u)\not=0$ for any $u\in \overline \D$; 
\item for any $w=(e^{i\omega_1},...,e^{i\omega_{n-1}})$, the function 
$u\in \D \mapsto F_w^*(u)/F_w(u)$ is an inner function in the Hardy space 
$H^2(\D)$ and the $\delta_F$ Schur parametrized (hence called functional) coefficients $w\mapsto \gamma_k(w)$, $k=0,...,\delta_F-1$ satisfy 
$|\gamma_k(w)|<1$ for any $w=(e^{i\omega_1},...,e^{i\omega_{n-1}})$.       
\end{itemize} 
Such an equivalence allows to transpose the $n$-dimensional problem to
\newline 
the $(n-1)$-dimensional setting. Such an approach is efficient for small values of $n$ ($n=2,3$,...) but difficult to implement in higher dimensions.   
\vskip 2mm
\noindent 
Therefore, the problem of testing the BIBO stability of a discrete $n$-rational filter appeals to investigate new strategies leading to criteria easier to implement in high dimensions from the computational point of view. Also, being able to decide whether the denominator $F$ of the transfer function vanishes on $\{(e^{i\theta_1},...,e^{i\theta_n})\,;\, 
\theta \in (\R/(2\pi \Z))^n\}$ seems to be an important challenge, since such an hypothesis is required prior to the formulation of any criterion for BIBO stability.     
\vskip 2mm
\noindent 
We intend in this paper to introduce a novel approach, based on the notion of am\oe ba of an algebraic hypersuface in $\T^n=(\C^*)^n$.    
The notion of amoeba was introduced by Gelfand, Krapanov and Zelevinsky in 1994 in their pionneer book on multidimensional determinants \cite{GKZ}.
The am\oe ba $\mathcal A_F$ of $F$ (one should better say of the zero set $F^{-1}(\{0\})$ of $F$ in $\T^n$) when $F$ is a Laurent polynomial in $n$ variables (in particular a polynomial in $n$ variables) is the image of $F^{-1}(\{0\})$ under the logarithmic map ${\rm Log}~: z 
\mapsto (\log |z_1|,...,\log |z_n|)$. 
Section \textsection 2 provides an overview of what is actually known about such concept, in view of the role it could play in relation with  BIBO stability. In section \textsection 3, we will formulate a criterion within the frame of am\oe ba and enlarge the notion of BIBO {\it stability} into that of {\it weak} BIBO {\it stability}. In section 
\textsection 4, we will analyze our approach from the algorithmic point of view.    

\section{Am\oe ba and related concepts, an overview} 

Let 
\begin{equation}\label{sect2-eq0} 
F = \sum\limits_{\alpha \in {\rm Supp}\, F \subset \Z^n}
c_\alpha X^\alpha \in \C[X_1^{\pm 1},...,X_n^{\pm 1}] \quad (c_\alpha \in \C^*,\ X^\alpha := X_1^{\alpha_1} \cdots X_n^{\alpha_n}) 
\end{equation} 
be a Laurent polynomial with zero set $F^{-1}(\{0\})$ in 
the complex torus $\T^n$. Let $\Delta (F)$ be the Newton polyhedron 
of $F$, that is the closed convex envelope in $\R^n$ 
of the set ${\rm Supp}\, F := \{\alpha \in \Z^n\,;\, c_\alpha \not=0\}$.     
We will always suppose that $F$ is a {\it true} Laurent polynomial in 
$n$ variables, which means that $\dim_{\R^n} \Delta(F)$, that is the dimension of the affine subspace of $\R^n$ generated by $\Delta (F)$ is maximal, that is equal to $n$.  
\vskip 2mm
\noindent 
Let $\mathcal A_F$ be the am\oe ba of $F$, that is the closed image 
${\rm Log} (F^{-1}(\{0\})\subset \R^n$ of the map 
$$
z\in F^{-1}(\{0\}) \subset \T^n 
\stackrel{\rm Log}{\longmapsto} (\log |z_1|,...,\log |z_n|) \in \R^n. 
$$
From the geometric point of view, the complement $\R^n \setminus \mathcal A_F$ is a $1$-convex open subset of $\R^n$, which amounts to say that its open connected components $E_\iota$ are convex. The number of such open connected components is bounded by the number of points in $\Delta (F)\cap \Z^n$ \cite{FPT00}. Let $\mathscr E_F$ be the finite set which elements are such components $E$.  
Each ${\rm Log}^{-1}(E)$, where $E\in \mathscr E_F$, is a Reinhardt domain in $\T^n$, that is a subdomain which is invariant under the pointwise multiplicative action of the real torus $\T_\R^n = \{(e^{i\theta_1},...,e^{i\theta_n})\,;\, 
\theta \in (\R/(2\pi\Z)^n\}$). Moreover it can be described as the maximal domain of convergence of a unique Laurent series $\sum_{k\geq 0} \gamma_{E,k} z^{\alpha_{E,k}}$, $\alpha_{E,k}\in \Z^n$, which sum represents $z\mapsto 1/F(z)$ in ${\rm Log}^{-1}(E)$. A key point is that {\it there is in fact a bijection between the finite set $\mathscr E_F$ and the family of all possible Laurent developments (with domains of convergence precisely ${\rm Log}^{-1}(E)$ for $E\in \mathcal E_F$) for $1/F$ along the monomials $z^\alpha$ for 
$\alpha \in \Z^n$} \cite{FPT00}. 

\begin{remark}[the case $n=1$]\label{sect2-rem1}{\rm In the case where $n=1$, the am\oe ba $\mathcal A_F$ consists in a finite number of points $-\infty< \log |a_1| < \cdots < \log |a_N| < +\infty$ ($a_j\in \C^*$) on the real line and each of the $N+2$ circular domains  
$$
\{z\in \C^*\,;\, |z|<|a_1|\},..., 
\{z\in \C^*\,;\, |a_j|<|z|<|a_{j+1}|\},...,\{z\in \C^*\,;\, |z|> |a_N|\}, 
$$
is the domain of convergence of a Laurent series which sum represents
$1/F$ in the corresponding domain. The domain $C=\{z\in \C^*\,;\, |z|> |a_N|\}$ is, among such list, the only one for which the associated Laurent development in $C$ is of the form $1/F(z) = \sum_{k\geq -M} \gamma_{C,k} z^{-k}$ for some $M\in \Z$ and hence the sequence $(\gamma_{C,k})_{k\in \Z}$ can be interpreted as the impulse response of a rational (realizable) discrete $1$-dimensional filter. 
} 
\end{remark} 

\noindent 
One can associate \cite{PRull08} 
to each $E\in \mathscr E_F$ a multiplicity $\nu_E = (\nu_{E,1},...,\nu_{E,n})\in \Delta(F)\cap \Z^n$, where $\nu_{E,j}$ is the degree of the loop 
\begin{equation}\label{sect2-eq1} 
\theta_j \in \Z/(2\pi \Z) 
\longmapsto F(\zeta_{E,1},\dots,\zeta_{E,j} e^{i\theta_j},\dots,\zeta_{E,n}) 
\end{equation} 
when $(\zeta_{E,1},...,\zeta_{E,n})$ is an arbitrary point in $E$, the degree of the loop \eqref{sect2-eq1} being independent on the choice of such point $\zeta_E$ in $E$. 
\vskip 2mm
\noindent 
For each point $\alpha \in \Delta(F) \cap \Z^n$, there is {\it at most} 
one component $E_\alpha\in \mathscr E_F$ such that $\nu_{E_\alpha} = \{\alpha\}$. 
If such is the case, let $\sigma_\alpha$ be the unique face of $\Delta(F)$ which contains $\alpha$ in its relative interior or (if no such face exists) equals $\{\alpha\}$~: for example, if $\alpha$ is a vertex of $\Delta(F)$, $\sigma_\alpha = \{\alpha\}$, while when $\alpha$ lies in the interior of $\Delta(F)$, $\sigma_\alpha = \Delta(F)$, etc. Then the cone 
$$
\Gamma_\alpha= \big\{x\in \R^n\,;\, \sigma_\alpha 
= \{\xi \in \Delta (F)\,;\, 
\langle \xi,x\rangle =\max_{u\in \Delta(F)} \langle u,x\rangle\}\big\}
$$
is the {\it recession cone} of $E_\alpha$, that is the largest cone $\Gamma$ of $\R^n$ such that $E_\alpha + \Gamma \subset E_\alpha$.  
Such recession cone equals $\{\0\}$ whenever $\alpha$ lies in the interior of $\Delta(F)$, hence the corresponding component $E_\alpha$ is, if it exists, bounded in this case. When $\alpha$ belongs to the boundary of 
$\Delta(F)$, the dimension of the recession cone is maximal (thus equal to $n$) if and only if $\alpha$ is a vertex of $\Delta(F)$. If $\alpha$ is a point of $\partial \Delta(F)\cap \Z^n$ which is not a vertex of 
$\Delta(F)$, then, if $E_\alpha$ exists, it is unbounded and its recession cone has dimension between $1$ and $n-1$. A major point is that {\it any vertex $\alpha$ of $\Delta(F)$ is the multiplicity $\nu_{E_\alpha}$ of a unique unbounded component $E_\alpha$ which admits as recession cone the cone 
\begin{equation}\label{sect2-eq1bis} 
\Gamma_\alpha= \big\{x\in \R^n\,;\, \{\alpha\}  
= \{\xi \in \Delta (F)\,;\, 
\langle \xi,x\rangle =\max_{u\in \Delta(F)} \langle u,x\rangle\}\big\}. 
\end{equation}
}    
Thus the cardinal of $\mathscr E_F$ lies between the number or vertices of $\Delta(F)$ and the cardinal of $\Delta(F)\cap \Z^n$; the number of bounded components in $\mathcal E_F$ (called the {\it genus} of the am\oe ba) lies between $0$ and the number of points in $\mathring \Delta(F) 
\cap \Z^n$ \cite{FPT00}.  
\vskip 2mm
\noindent 
An important concept related to the am\oe ba $\mathcal A_F$ is that of 
{\it contour}. Let $(F^{-1}(\{0\}))_{\rm sing}$ be the subvariety of singular points of the algebraic hypersurface $F^{-1}(\{0\})$~: if 
$F$ is assumed to be irreducible in $\C[X_1^{\pm 1},...,X_n^{\pm 1}]$, 
$(F^{-1}(\{0\}))_{\rm sing}$ is defined 
as 
$$
(F^{-1}(\{0\}))_{\rm sing} = 
\big\{z\in \T^n\,;\, 
F(z) = \frac{\partial F}{\partial z_1}(z) = \cdots = 
\frac{\partial F}{\partial z_n}(z)=0\big\} ; 
$$
if $F=F_1^{q_1}\cdots F_M^{q_M}$ is the decomposition of $F$ in irreducible factors in $\C[X_1^{\pm 1},...,X_n^{\pm 1}]$ 
$$
(F^{-1}(\{0\}))_{\rm sing} = 
\bigcup\limits_{j=1}^M (F_j^{-1}(\{0\}))_{\rm sing} 
\cup \bigcup\limits_{\stackrel{1\leq j,k\leq M}{j\not=k}}
(F_j^{-1}(\{0\}) \cap F_k^{-1}(\{0\})). 
$$
In any case, the codimension of $(F^{-1}(\{0\}))_{\rm sing} 
\subset F^{-1}(\{0\})$ in $\T^n$ is at least equal to $2$. 
One denotes as $(F^{-1}(\{0\}))_{\rm reg}$ the (in general non closed) submanifold of $\T^n$ defined as $F^{-1}(\{0\}) \setminus (F^{-1}(\{0\}))_{\rm sing}$.   

\begin{definition}[contour of $\mathcal A_F$ \cite{PT05}] The contour of the am\oe ba $\mathcal A_F$ is the union of ${\rm Log} (F^{-1}(\{0\}))_{\rm sing}$ with the set of critical values of 
${\rm Log}~: (F^{-1}(\{0\}))_{\rm reg} \rightarrow \R^n$.
\end{definition} 

\noindent
The contour of $\mathcal A_F$ contains necessarily the boundary of 
$\mathcal A_F$. A major result concerning the description of the contour of $\mathcal A_F$ is the following result due to G. Mikhalkin \cite{Mik00,PRull08}. 

\begin{theorem}[\cite{Mik00,PRull08}, see also \textsection 4.1 in 
\cite{Y12}]\label{sect2-thm1}  
Suppose that $F$ is a reduced Laurent polynomial, that is $F=F_1\cdots F_M$ where each $F_j$ is irreducible. Let 
\begin{equation}\label{sect2-GaussMap} 
z\in (F^{-1}(\{0\}))_{\rm reg}  
\longmapsto \gamma_F(z) := \Big[ 
z_1 \frac{\partial F}{\partial z_1}(z) : \cdots : 
z_n \frac{\partial F}{\partial z_n}(z)\Big] \in \P^{n-1}(\C)  
\end{equation} 
be the Gauss logarithmic map. One has 
\begin{equation}\label{sect2-eq2} 
{\rm contour}\, (\mathcal A_F) = 
\overline{{\rm Log}\, \big(
\gamma_F^{-1} (\P^{n-1}(\R)\big)},  
\end{equation}  
where $\P^{n-1}(\R)= (\R^n\setminus \{\underline 0\})/\R^*$ denotes the real $(n-1)$-dimensional projective space.  
\end{theorem}  

\noindent 
From this theorem, an algebraic algorithm based on elimination theory can be realized in order to compute the contour of $\mathcal A_F$ for bivariate polynomials $F(X_1,X_2) \in \Z[X_1,X_2]$ (see \cite[Algorithm 3]{BKS16}). This algebraic algorithmic procedure is based on the construction we sketch below. Let $u\in \R$ be a real parameter. Consider the two polynomials 
\begin{equation}\label{sect2-eq4}
P(X_1,X_2,u) = F(X_1,X_2),\quad Q(X_1,X_2,u) = \frac{\partial F}{\partial X_1} 
(X_1,X_2) + u \frac{\partial F}{\partial X_2}(X_1,X_2) 
\end{equation}
in $\Z[u][X_1,X_2]$. One can compute formally in an exact way the Sylvester resultant $R_{X_2}(u,X_1)\in \Z[u,X_1]$ of $P$ and $Q$ considered as elements of 
$\Z[u,X_1] [X_2]$ and the Sylvester resultant $R_{X_1}(u,X_2) \in \Z[u,X_2]$ of 
$P$ and $Q$ considered this time as elements of $\Z[u,X_2][X_1]$. See for example \cite{Lang} for the construction of the Sylvester resultants of two polynomials in $\mathbb A[Z]$ where $\mathbb A$ is a commutative domain of integrity with fraction field $\mathbb K$, as $\mathbb A=\Z[u,X_1]$ 
(and $Z=X_2$) or 
$\mathbb A = \Z[u,X_2]$ (and $Z=X_1$) in our case. Theorem \ref{sect2-thm1} implies that $(x_1,x_2)\in {\rm contour}(\mathcal A_F)$ if and only if there is at least one point $u\in \P^1(\R)$ such that $P(X,Y,u)$ and $Q(X,Y,u)$ share a common zero $(z_1,z_2)\in \T^2$ which lies in the orbit ${\rm Log}^{-1}(\{(x_1,x_2)\})$, that is such that $|z_1|= e^{x_1}$ and 
$|z_2| = e^{x_2}$. Elimination theory implies that $(z_1,z_2)$ satisfies
\begin{equation}\label{sect2-eq5} 
R_{X_2} (u,z_1) = R_{X_1}(u,z_2) = 0. 
\end{equation} 
Given $u\in \Q = \P^1(\Q) \setminus \{[1:0]\} = \P^1(\Q) \setminus \{\infty\}$, one can compute exactly thanks to Newton's method the (at most $\deg_{X_1} R_{X_2} \times \deg_{X_2} R_{X_1}$) pairs of algebraic numbers $(z_1,z_2) \in (\overline \Q)^2 \cap \T^2$ which satisfy \eqref{sect2-eq5}.   
Then, one can extract from this list the sublist of pairs of points 
which satisfy 
\begin{equation} \label{sect2-eq6} 
P(z_1,z_2,u) = Q(z_1,z_2,u) = 0,  
\end{equation}    
that is such that $(\log |z_1|,\log |z_2|) \in {\rm contour} (\mathcal A_F)$. Repeating this procedure for $u= -N + 2 \ell N/M $, 
$\ell =0,...,M-1$, where $M >>N >>1$, leads to a construction of ${\rm contour}(\mathcal A_F)$ (see \cite[Algorithm 3]{BKS16} for the numerical code under \texttt{Matlab} or also \cite[\textsection 5]{Oss19} for the formal code under \texttt{Sage}). 
\vskip 2mm
\noindent 
As a consequence of this method, one can state the following proposition. 
\begin{proposition}[is $(0,0)$ in the contour of $\mathcal A_F$~?]\label{sect2-prop1} 
Let $F\in \Z[X_1,X_2]$ be an irreducible polynomial in $2$ variables such that $\dim_{\R^2} \Delta(F)=2$. There is an exact procedure based on Schur-Cohn test to decide whether $(0,0)$ lies (or not) in the contour of $\mathcal A_F$.
\end{proposition} 

\begin{proof} Let us form the resultants $R_{X_2}(u,X_1)\in \Z[u][X_1]$, $R_{X_1}(u,X_2) \in \Z[u][X_2]$. Consider $u$ as a real parameter.
The Schur-Cohn test allows to decide for which possible values of the real parameter $u\in \overline \Q$ the polynomials $R_{X_1}(u,X_2)$ and 
$R_{X_2}(u,X_1)$ may both have a root (necessarily in $\overline \Q$) on the unit circle $|\zeta[=1$ of the complex plane. We are then left with a finite number of situations to test in order to decide whether the two polynomials $P(X_1,X_2,u)$ and $Q(X_1,X_2,u)$ defined in \eqref{sect2-eq4} have at least both a root on the unit circle, which means in this case that $(0,0)$ belongs to the contour of $\mathcal A_F$.     
\end{proof}       
\vskip 2mm
\noindent 
Another important concept related to $F$ which provides geometric information on $\mathcal A_F$ is the following convex function 
$$
R_f~: x\in \R^n \longmapsto \int_{\T_\R^n} 
F(e^{x_1 + i\theta_1},...,e^{x_n+i\theta_n})\, d\nu_{\T_\R^n} (\theta), 
$$
where $\T_\R^n = (\R/(2\pi\Z))^n$ equipped with its normalized Haar measure $d\nu_{\T^n_\R}$. It was introduced by L. Ronkin in 
\cite{Ronk00} and is thus called the {\it Ronkin function} of $F$. The three important facts to retain about such function are the following. 
\begin{enumerate} 
\item The function $R_F$ is affine in the connected component (of $\R^n \setminus \mathcal A_F$) $E_\alpha\in \mathcal E_F$ with multiplicity $\alpha \in \Delta(F) \cap \Z^n$, provided of course such component exists \cite{PRull08}. More precisely, 
\begin{equation} 
\label{sect2-eq7}  
\forall x\in E_\alpha,\quad R_F(x) = \rho_\alpha + \langle \alpha,x\rangle. 
\end{equation} 
\item 
When $\alpha$ is a vertex of $\Delta(F)$ (and hence $E_\alpha$ exists,  
with $n$-dimensional recession cone given by \eqref{sect2-eq1bis}), then 
$\rho_\alpha =\log |c_\alpha|$, where $c_\alpha$ is the coefficient of $X^\alpha$ in the developped expression \eqref{sect2-eq0} for $F$.
\item The singular support of the distribution $\Delta ([R_F])$ (where $\Delta$ is the Laplace operator and $[R_F]$ means that $R_F$ is considered in the sense of distributions) is contained in the contour of $\mathcal A_F$ \cite[Theorem 3.1]{Oss19}.  
\end{enumerate}
Although $R_F$ is just a continuous function inside $\mathcal A_F$, one can compute numerically when $n=2$ the Laplacian of the associate distribution $[R_F]$ (see \cite[\textsection 5]{Oss19}). The main reason why such 
a method works is that the singularities of $\log |F|$ on $\T^n$ are gentle ones (the function $\log |F|$ is locally integrable on $\T^n$) and the use of the Laplace operator (because of its symmetric form 
with respect to coordinates) is a basic (primitive) tool for the detection of contours in image processing. Such numerical computation provides (unfortunately in some empiric way) a suprizingly convincing picture both of the am\oe ba and simultaneously of its contour \cite[\textsection 5]{Oss19}.  
\vskip 2mm
\noindent
The convex Ronkin function $R_F$ has a companion $p_F$ 
which is much easier to describe since it is realized in $\R^n$ as the upper envelope of a finite number of affine functions with slopes in $\Z^n$, hence can be interpreted as the evaluation function in $\R^n$ of a {\it tropical Laurent polynomial} since the operations 
$$
(a,b)\in ([-\infty,+\infty[)^2\longmapsto \max (a,b),\quad (a,b) \in ([-\infty,+\infty[)^2
$$
substitute to the usual addition and multiplication in the (tropical) 
{\it max-plus calculus}. Let $\nu~: \mathscr E_F \rightarrow \Delta(F) \cap \Z^n$ the multiplicity map which associates to each $E\in \mathscr E_K$ its multiplicity $\nu_E$ ; then $p_F$ is defined as 
$$
p_F~: x\in \R^n \longmapsto \max\limits_{\alpha \in  {\rm Im}\, \nu} 
(\rho_\alpha + \langle \alpha,x\rangle). 
$$  
One has that $p_F(x)\leq R_F(x)$ for any $x\in \R^n$ and $p_F(x)=R_F(x)$ in $\R^n \setminus \mathcal A_F$. For each $\alpha \in {\rm Im}\, \nu$, let $C_\alpha$ be the $n$-dimensional convex polyhedron (possibly unbounded) of $\R^n$ defined as 
$$
C_\alpha = \{x\in \R^n\,;\, p_F(x) + \check p_F(\alpha) = \langle \alpha,x\rangle\},  
$$
where $p_F~: 
\xi\in \R^n \mapsto \sup_{x\in \R^n} (\langle \xi,x\rangle - p_F(x))\in ]-\infty,+\infty]$ is the Legendre transform of $p_F$, which  satisfies 
$\check p_F^{-1}(\{+\infty\})=\Delta(F)$ \cite{PRull08}. The interiors $\mathring C_\alpha$, $\alpha \in {\rm Im}\, \nu$ are pairwise disjoints, and the the complement of their union equals the set of critical values of $p_F$, that is the subset of points in $\R^n$ about which $p_F$ is not an affine map. Since $p_F$ and $R_F$ coincide on $\R^n \setminus \mathcal A_F$, one has $E_\alpha \subset C_\alpha$ for any $\alpha\in {\rm Im}\, \nu$. Observe then that, given a point $x\in \R^n \setminus \mathcal A_F$, in order to decide to which component $E_\alpha$ it belongs, one needs to check to which $\mathring C_\alpha$ it belongs.  
\vskip 2mm
\noindent  
Let for any $N\in \N^*$ (in particular $N=2^k$, $k\in \N$) $\F_N$ be the multiplicative group of $N$-roots of unity and 
$$
F_N(X) = \prod\limits_{\varpi\in \F_N^n} 
F(\varpi_1 X_1,...,\varpi_N X_N). 
$$
An iterative procedure to compute the $F_{2^k}$ ($k\in \N$) inspired by the Gauss-Cooley-Tukey FFT algorithm has been proposed in \cite[\textsection  3]{FMMW19}.  
Observe that $F$ and $F_N$ share the same am\oe ba $\mathcal A_F$ for all $N\in \N^*$. It follows also from Galois theory that $F_N \in \Z[X_1^{\pm 1},...,X_n^{\pm 1}]$ as soon as $F\in \Z[X_1^{\pm 1},...,X_n^{\pm 1}]$. Since the integral of a continuous function can be approximated by Riemann sums, the Ronkin function $R_F$ admits in $\R^n \setminus \mathcal A_F$ the uniform approximation on any compact subset 
\begin{equation}\label{sect2-eq9} 
R_F (x) = \lim\limits_{N\rightarrow +\infty} 
\frac{R_{F_N}(x)}{N^n}\quad (x\in \R^n \setminus \mathcal A_F). 
\end{equation}   
In order to exploit such idea, K. Purbhoo introduced the non-archimedean concept of (tropically) {\it lopsided am\oe ba of $F$}. We recall here this construction. 
\vskip 1mm
\noindent 
Recall that a finite set $\{\tau_\iota\,;\, \iota \in I\}$ of strictly positive real numbers (with possible repetitions) is said to be (tropically) {\it lopsided} if and only if there is a (necessarily) unique index $\iota_0\in I$ such that 
$$
\tau_{\iota_0} > \sum\limits_{\iota \in I \setminus \{\iota_0\}}
\tau_\iota. 
$$

\begin{definition}[lopsided am\oe ba of $F$ \cite{Purb08}] 
The {\rm lopsided am\oe ba} $\mathscr L_F$ of the Laurent polynomial $F(X)=\sum_{\alpha \in {\rm Supp} F\subset \Z^n} 
c_\alpha X^\alpha$ is the image by ${\rm Log}$ of the subset of $\T^n$ which consists in the set of points $z=(z_1,...,z_n)\in \T^n$ where the 
set $\{|c_\alpha|\, |z^\alpha|\,;\, \alpha \in {\rm Supp} F\}$ of strictly positive numbers is \underline{not} lopsided.    
\end{definition} 

\noindent One has necessarily that $\mathcal A_F\subset \mathcal L_F$ since if $x\in \mathscr A_F$, it is clearly impossible for the set 
$\{|c_\alpha| e^{\langle \alpha,x\rangle}\,;\, \alpha \in {\rm Supp}\, F\}$ to be lopsided. K. Purbhoo observed in \cite{Purb08} the following. 

\begin{theorem}[\cite{Purb08}, see also \textsection 2.2 in \cite{Y12}]\label{sect2-thm2} Suppose that $F$ is a Laurent polynomial in $n$ variables such that $\dim_{\R^n}\Delta(F)=n$. For any $\varepsilon >0$, one can find $N_\varepsilon \in \N^*$ such that 
\begin{equation}\label{sect2-eq10} 
\forall N\geq N_\varepsilon,\quad 
{\rm dist} \big(\mathscr L_{F_N},\mathcal A_F\big)< \varepsilon. 
\end{equation} 
\end{theorem} 

\noindent 
This result can be quantified as follows. Let 
$$
c_F = \max\limits_{1\leq j\leq n} \big(\sup\limits_{\xi =(\xi_1,...,\xi_n)\in \Delta(F)} \xi_j  -  \inf_{\xi= (\xi_1,...,\xi_n) \in \Delta(F)} \xi_j\big).
$$ 
Let $d_F = \sup_{t\in \N} (E_{\Delta(F)}(t)/t^n)$, where $t\mapsto E_{\Delta(F)}(t)$ is the {\it Ehrhart polynomial} of the $n$-dimensional Newton polyhedron 
$\Delta(F)$ (see \cite{BtMc85,Beck05} for the definition and properties of the Ehrhart polynomial of a convex polyhedron such as $\Delta(F)$ and \cite{BtMc85} for an upper estimate of $d_F$). Then, if $x\in \R^n$ is such that 
${\rm dist}(x,\mathcal A_F)\geq \varepsilon>0$, then for any $N$ such that 
\begin{equation}\label{sect2-eq11} 
N \geq \frac{1}{\varepsilon} \Big( (n^2-1) \log N + \log \frac{16 c_F d_F}{3}\Big), 
\end{equation} 
the point $x$ cannot belong to the lopsided am\oe ba $\mathscr L_{F_N}$ \cite{Purb08}. 
\vskip 2mm
\noindent
One can complete the information given by Theorem \ref{sect2-thm2} and quantified as in \eqref{sect2-eq11} with the following companion proposition. 

\begin{proposition}[\cite{Purb08,Y12}, see also 
\textsection 3.1.5 in \cite{Y12}]\label{sect2-prop3}
Let $\alpha \in \Delta(F)\cap \Z^n$ such that a component $E_\alpha$ such that $\nu(E_\alpha)=\alpha$ exists in $\mathscr E_F$ and $x\in E_\alpha$ with 
$d(x,\mathcal A_F)\geq \varepsilon$. Then, for $N\geq N_\varepsilon$  
such that \eqref{sect2-eq11} holds, the leading term 
$|c_{\beta}|\, e^{\langle \beta,x\rangle}$ in the lopsided 
finite set $\{|c_{\alpha_N}|\, e^{\langle \alpha_N,x\rangle}\,;\, 
\alpha_N \in {\rm Supp} F_N\}$, where $F_N=\sum_{\alpha_N 
\in {\rm Supp} F_N} c_{\alpha_N} X^{\alpha_N}$, is such that 
$\beta = N^n\, \alpha$.      
\end{proposition} 
               
\section{BIBO stability and am\oe bas} 
 
Let $S$ be a discrete $n$-linear time invariant system with transfer function the rational function 
\begin{equation}\label{sect3-eq1} 
\frac{B(X_1^{-1},...,X_n^{-1})}{A(X_1^{-1},...,X_n^{-1})} = X^\gamma \frac{G(X_1,...,X_n)}{F(X_1,...,X_n)}\in \C(X_1,...,X_n), 
\end{equation} 
where 
$\gamma \in \Z^n$ and $G,F\in \C[X_1,...,X_n]$ are coprime in $\C[X_1,...,X_n]$, both $F$ and $G$ being coprime with $X_1\cdots X_n$. 
\vskip 2mm
\noindent 
Our first observation is that the condition that $z\mapsto G(z)/F(z)$ is regular about 
$$
{\rm Log}^{-1}(\{\underbar 0\}) = 
\{z = (e^{i\theta_1},...,e^{i\theta_n})\,;\, \theta \in (\R/(2\pi\Z))^n\}
$$
is equivalent to the fact that its polar set $F^{-1}(\{0\})$ in $\T^n$ does not intersect ${\rm Log}^{-1}(\{\underbar 0\})$, which amounts to say that $\underbar 0 \in \R^n\setminus \mathcal A_F$.   
One can then state the following result.  

\begin{theorem}\label{sect3-thm1} 
Suppose that $\underline 0\in \R^n \setminus \mathcal A_F$, where 
$\dim_{\R^n} \Delta(F)=n$. 
A necessary and sufficient condition for a discrete $n$-rational filter $S$ with the rational function \eqref{sect3-eq1} as transfer function to be (strongly) {\rm BIBO} stable is that $\xi = {\boldsymbol 0} \in {\rm Supp}\, F$ and $x=\underline 0 \in E_{\boldsymbol 0}$. 
\end{theorem} 

\begin{remark}[why {\it strong} BIBO stability ?] {\rm We speak here about {\it strong} BIBO stability (which is the usual notion as described up to now) in order to differentiate it weaker one that we will introduce next in Definition \ref{sect2-def2}.}
\end{remark}  

\begin{proof}
Suppose that $S$ is BIBO stable. Let 
$\overline \D^n = \{z\in \C^n\,;\, |z_1|\leq 1,...,|z_n|\leq 1\}$. 
Since 
$
\{z=(z_1,...,z_n)\in \T^n\,;\, 0<|z_1|\leq 1,...,0<|z_n|\leq 1\}
$
equals the Reinhardt domain ${\rm Log}^{-1} \big(\{x=(x_1,...,x_n)\in \R^n\,;\, 
x_1\leq 0,...,x_n\leq 0\}\big)$ and $F$ is coprime with $X_1\cdots X_n$, 
it is equivalent to say that $F$ does not vanish in $\overline \D^n$ 
and that the cone $\boldsymbol \Gamma^- := \{x=(x_1,...,x_n)\in \R^n\,;\, 
x_1\leq 0,...,x_n\leq 0\}$ lies entirely in some connected component 
$E= E_\alpha\in \mathscr E_F$, where $\alpha$ is necessarily a vertex of 
$\Delta(F)$ since the cone $\boldsymbol \Gamma^{-}$ is $n$-dimensional. 
It follows from the fact that ${\rm Supp}\, F \subset \N^n$ ($F$ is a polynomial in $X_1,...,X_n$) that the only possible vertex of $\Delta(F)$ for such a situation to occur is that $\alpha = \boldsymbol 0$, which implies that $\boldsymbol 0 \in {\rm Supp}\, F$. Since $\boldsymbol \Gamma^- \subset E_{\boldsymbol 0}$, one has in particular that 
$x=\underline 0 \in E_{\boldsymbol 0}$. Conversely, suppose $\boldsymbol 0 \in {\rm Supp}\, F$ and $\underline 0\in E_{\boldsymbol 0}$. The recession cone (see \eqref{sect2-eq1bis})
$$
\Gamma_{\boldsymbol 0}= \big\{x\in \R^n\,;\, \{\boldsymbol 0\}  
= \{\xi \in \Delta (F)\,;\, 
\langle \xi,x\rangle =\max_{u\in \Delta(F)} \langle u,x\rangle\}\big\} 
$$
of the unbounded connected component $E_{\boldsymbol 0}$ of 
$\R^n \setminus \mathcal A_F$ contains $\boldsymbol \Gamma^-$, as it is immediate to check. This implies that $z\mapsto 1/F(z)$ is holomorphic 
in the Reinhardt domain $\{z=(z_1,...,z_n)\in \T^n\,;\, 0<|z_j|<1\}$, which means that $F$ does not vanish there. Since $F$ is coprime with $X_1\cdots X_n$, $F$ cannot vanish on the union of the coordinate axis either, which implies that the $n$-rational filter $S$ is BIBO stable.        
\end{proof} 
\vskip 2mm
\noindent 
Theorem \ref{sect3-thm1} suggests to introduce two concepts related to the BIBO stability property. 

\begin{definition}[BIBO stability domain]\label{sect3-def1} 
Let $S$ be a discrete $n$-rational filter with a rational function \eqref{sect3-eq1} (with its properties, together with the condition $\dim_{\R^n} \Delta(F)=n$) as transfer function. If 
$\boldsymbol 0 \in {\rm Supp}\, F$, the connected component $E_{\boldsymbol 0}$ of $\R^n \setminus \mathcal A_F$ is called the 
{\rm BIBO stability domain} of the $n$-filter $S$. If 
$\boldsymbol 0 \notin {\rm Supp}\, F$, one decides that the {\rm BIBO} stability domain of $S$ is empty.        
\end{definition} 
 
\begin{definition}[BIBO weak stability]\label{sect2-def2}  Let $S$ be a discrete $n$-rational filter with a rational function \eqref{sect3-eq1} as in Definition \ref{sect3-def1}. The discrete $n$-rational filter $S$ is said to be {\rm BIBO weakly stable} if $\boldsymbol 0\in {\rm Supp}\, F$ and 
$\underline 0$ belongs to the topological boundary of the connected component $E_{\boldsymbol 0}$, which is part of the contour of $\mathcal A_F$. If it is the case, the component $E_{\boldsymbol 0}$ is called the {\rm weak BIBO stability domain} of the discrete $n$-rational filter $S$ ; otherwise the weak {\rm BIBO} domain of $S$ is considered as empty. 
\end{definition} 

\section{Algorithmic considerations} 

In this section, one considers a polynomial 
$F\in \Z[X_1,...,X_n]$ such that $\dim_{\R^n} \Delta(F) =n$, namely 
and $\boldsymbol 0\in {\rm Supp}\, F$, namely 
$$
F(X_1,...,X_n) = \sum\limits_{\alpha \in {\rm Supp}\, F
\subset \N^n} c_\alpha X^\alpha \quad (c_\alpha \in \Z^*,\ \boldsymbol 0 
\in {\rm Supp}\, F). 
$$   
All $F_N$ for $N\geq 1$ (in particular $N=2^k$ for $k\in \N$) remain in 
$\Z[X_1,...,Z_n]$.
\vskip 2mm
\noindent 
In order to state results from the algorithmic point of view, it is important to precise with which precision real or complex quantities are 
evaluated. Let us fix $\varepsilon_0>0$ as the threshold error. 
\vskip 2mm
\noindent 
Consider the assertion {\bf ($\bfA_{2^{-M_0}}$)}~: ``the distance of $\underline 0$ to 
$\R^n \setminus \mathcal A_F$ is at least equal to $2^{-M_0}$''. 
Then, we know from \eqref{sect2-eq11}, together with the precisions given by Proposition \ref{sect2-prop3}, that, as soon as 
$$
2^k \geq 2^{M_0}\, \Big( (n^2-1) k \log 2  + \log \frac{16 c_F d_F}{3}\Big),
$$ 
then, if 
$$
F_{2^k}(X) = \sum\limits_{\alpha_k \in {\rm Supp}\, F_{2^k}} 
c_{2^k,\alpha_k} X^{\alpha_k} \quad (c_{2^k,\alpha_k} \in \Z^*),  
$$
the set $\{|c_{2^k,\alpha_k}|\,;\, \alpha_k \in {\rm Supp}\, F_{2^k}\}\subset \N^*$ is lopsided, with leading term among the set $\{
|c_{2^k,2^{nk} \alpha}|\,;\, \alpha \in {\rm Im}\, \nu\}$. 
This is true as soon as $k\geq M_0 + \gamma_F$, where $\gamma_F$ is an a positive constant depending on $\Delta(F)$. Therefore, one can proceed algorithmically as follows in order (if possible) to validate the assertion {\bf ($\bfA_{2^{-M_0}}$)} :  
\begin{enumerate} 
\item compute the list of coefficients of the polynomial $F_{2^k}$ iteratively up to $k=M_0+\gamma_F$ (using the algorithmic procedure introduced in \cite{FMMW19}); 
\item  extract at each step $k$ the ${\rm card}\, ({\rm Im}\, \nu)$ strictly positive integer coefficients 
$c_{2^k,\alpha 2^k}$,  $\alpha \in {\rm Im}\, \nu$, from such list;  
\item test at each step $k$ whether one of the lobsided conditions 
\begin{equation}\label{sect4-eq1} 
|c_{2^k,\alpha 2^{kn}}| > \sum\limits_{\{\alpha_{k}\in {\rm Supp}\,F_{2^k}\,;\, \alpha_k\not= \alpha 2^{kn}\}} |c_{2^k,\alpha_k}|
\end{equation} 
is fulfilled (each such test being exact since $F_{2^k} \in \Z[X_1,...,X_n]$);  
\item if one of the above lopsided conditions is true, then the assertion 
{\bf ($\bfA_{2^{-M_0}}$)} is validated (observe that we also know then in which component $E_\alpha$ lies the point $\underline 0$) and one stops the procedure~; if not, the procedure goes on until $k=M_0+\gamma_F$. 
If it fails up to this point, it means that either 
$(0,0)\in \mathcal A_F$ or the threshold $2^{-M_0}$ is not sufficient to validate the assertion {\bf ($\bfA_{2^{-M_0}}$)}. 
\end{enumerate}
\vskip 1mm
\noindent 
If the above algorithmic procedure {\bf ($\bfA_{2^{-M_0}}$)} ends up with a validation, one can deduce a procedure to prove or disprove the assertion 
{\bf ($\bfA^0$)} : ``$\underline 0 \in E_{\boldsymbol 0}$''. 
One just need to analyze which $\alpha$ provides the lopsided 
inequality in \eqref{sect4-eq1} at step $3$. If it $\alpha=0$, then 
$(\bfA^0)$ is validated ; if it is $\alpha \not=0$, 
$(\bfA^0)$ is disproved. Thus, according to the fact that 
validation procedure for $(\bfA_{2^{-M_0}})$ concludes positively, 
we obtain in this way a test for BIBO stability of the discrete $n$-rational filter with transfer function such as \eqref{sect3-eq1}.    
\vskip 2mm
\noindent 
When $n=2$ and the validation procedure of {\bf ($\bfA_{2^{-M_0}}$)} fails, one can use \eqref{sect2-prop1} to prove or disprove exactly the assertion 
{\bf (B)}~: ``$\underline 0$ lies in the contour of $\mathcal A_F$''. 
Such procedure provides also the value of $\gamma_F(\zeta)\in \P^1(\R)$, where 
$\gamma_F$ denotes the logarithmic Gauss map introduced in \eqref{sect2-GaussMap} and $\zeta\in {\rm Log}^{-1}(\{0\}) = 
\{(e^{i\theta_1},e^{i\theta_2})\,;\, (\theta_1,\theta_2) 
\in (\R/(2\pi \Z))^2\}$ is such that $\gamma_F(\zeta) \in \P^1(\R)$. 
The value $\gamma_F(\zeta)$ stands for the normal complex direction to 
the smooth complex curve ${\rm log} ((F^{-1}(\{0\}))_{\rm reg})$ in the tubular domain $\R^2_x + i \R^2_\theta$ ($\log $ being here the 
multivalued function $z \mapsto {\rm Log} z + i 
({\rm arg} (z_1),{\rm arg}(z_2))$). The fact that $\gamma_F(\zeta) \in \P^1(\R)$ (in which case ${\rm Log} (\zeta) = (0,0)$ is in the contour of $\mathcal A_F$) can be interpreted as the fact that the direction $\gamma_F(\zeta)$ is ``horizontal'' in the  vertical strip 
$\R^2_x + i\R^2_\theta$. In case $(0,0)$ belongs to the boundary of $\mathcal A_F$ (which is a subset of the contour), such a direction $\gamma_F(\zeta)$ corresponds to a normal direction to the boundary of the am\oe ba $\mathcal A_F$ at the point $(0,0)$. 
\vskip 2mm
\noindent 
Let still $n=2$.  
Suppose assertion {\bf (B)} has been proved. Let now 
{\bf (C)} be the assertion: ``$\underline 0$ is a boundary point of $E_{\boldsymbol 0}$'' (that is the corresponding discrete rational $2$-filter is weakly BIBO stable in the sense of Definition \ref{sect2-def2}).   
The numerical procedures \texttt{RONKIN}, \texttt{AMIBE} (under 
\texttt{MATLAB}) and \texttt{ContourAmoeba} (under the environment of formal calculus \texttt{Sage}) proposed in \cite[\textsection 5]{Oss19} 
(see also \cite[Algorithm 3]{BKS16}) lead to a representation of 
$\mathcal A_F$ and its contour just by plotting the two-dimensional graph of $\Delta ([R_F])$. One cannot conclude from such routines to an algorithmic procedure from which one could validate the assertion ${\bf (C)}$ since $\Delta([R_F])$ is a distribution 
which is roughly numerically evaluated as a function. Nevertheless, the result of such algorithmic procedures (\texttt{RONKIN}, \texttt{AMIBE}, see \cite[\textsection 5]{Oss19}) 
allow to guess that $(0,0)$ is close to a point in the boundary of $E_{\boldsymbol 0}$ which is not a {\it branching point} for the contour of $\mathcal A_F$. If this is the case, the validation of 
assertion {\bf (B)} implies that of assertion {\bf (C)}. 
Note that disproving {\bf (B)} also disproves {\bf (C)} since the boundary of $\mathcal A_F$ is a subset of the contour of $\mathcal A_F$.          
We get in this way a test for weak BIBO stability when the test for BIBO stability {\bf ($\bfA^0$)} fails since {\bf ($\bfA_{2^{-M_0}}$)} fails.

\end{document}